\documentclass[11pt]{amsart}
\usepackage[top=1in,bottom=1in,left=1in,right=1in]{geometry}
\usepackage{amsmath}
\usepackage{amsfonts}
\usepackage{amsthm}
\usepackage[foot]{amsaddr}
\usepackage{mathrsfs}  
\usepackage{tikz}
\usepackage{subcaption}
\usepackage{algorithm,algorithmic}
\pgfdeclarelayer{bg}    
\pgfsetlayers{bg,main}
\newtheorem{theorem}{Theorem}

\newtheorem{lemma}{Lemma}
\newtheorem{observation}{Observation}

\newcommand{\rk}{\operatorname{rk}}
\DeclareMathAlphabet{\mathttit}{T1}{\ttdefault}{\mddefault}{\sldefault}
\newcommand{\poly}{\mathttit}
\newcommand{\per}{\operatorname{per}}
\newcommand{\transpose}{^{\mathsf T}}

\newcommand{\kC}{\textsc{$k$-Cycle}}
\newcommand{\keC}{\textsc{$k,e$-Cycle}}
\newcommand{\LC}{\textsc{$k$-Long Cycle}}
\newcommand{\keLC}{\textsc{$k,e$-Long Cycle}}
\newcommand{\SkCC}{\textsc{Shortest $k,t,e$-Colourful Cycle}}
\newcommand{\SkpCC}{\textsc{Shortest $k',t,e$-Colourful Cycle}}
\newcommand{\HC}{\textsc{Hamiltonian Cycle}}

\begin{document}

\title{Finding longer cycles via shortest colourful cycle}

\ifdefined\A
\author{Anonymous}
\else

\author{Andreas Bj\"orklund}
\address[AB]{IT University of Copenhagen} 
\email{anbjo@itu.dk}
\author{Thore Husfeldt}
\address[TH]{IT University of Copenhagen} 
\email{thore@itu.dk}
\date{}
\fi

\maketitle
\begin{abstract}
We consider the parameterised {\keLC} problem, in which you are given an $n$-vertex undirected graph $G$, a specified edge $e$ in $G$, and a positive integer $k$, and are asked to decide if the graph $G$ has a simple cycle through $e$ of length at least $k$. We show how to solve the problem in $1.731^k\operatorname{poly}(n)$ time, improving over the $2^k\operatorname{poly}(n)$ time algorithm by [Fomin~\emph{et al.}, TALG 2024], but not the more recent $1.657^k\operatorname{poly}(n)$ time algorithm by [Eiben, Koana, and Wahlstr\"om, SODA 2024]. 
When the graph is bipartite, we can solve the problem in $2^{k/2}\operatorname{poly}(n)$ time, matching the fastest known algorithm for finding a cycle of length \emph{exactly} $k$ in an undirected bipartite graph [Bj\"orklund \emph{et al.}, JCSS 2017].

Our results follow the approach taken by [Fomin~\emph{et al.}, TALG 2024], which describes an efficient algorithm for finding cycles using many colours in a vertex-coloured undirected graph. Our contribution is twofold. First, we describe a new algorithm and analysis for the central colourful cycle problem, with the aim of providing a comparatively short and self-contained proof of correctness. Second, we give tighter reductions from {\keLC} to the colourful cycle problem, which lead to our improved running times.
\end{abstract}

\newpage
\section{Introduction}
The \textsc{Hamiltonian Cycle} problem asks if a given undirected $n$-vertex graph $G$ has a simple cycle of length equal to its number of vertices.  The problem is on Karp's original list of NP-hard problems from 1972 \cite{Karp1972}. 
Bj\"orklund~\cite{Bjorklund2014} showed that there is a Monte Carlo algorithm asymptotically running in $O(1.657^n)$ time that detects whether an  undirected $G$ has a Hamiltonian cycle or not with high probability. 
In the present paper we show how the key ideas of that algorithm can be used also for a parameterised version of \textsc{Hamiltonian Cycle}.

One natural parameterisation of {\HC} is the {\kC} problem:
Given a graph, decide if it contains a simple cycle of length exactly $k$. It is known that the Hamiltonicity algorithm mentioned above exists in this parameterised version, showing that the  \textsc{$k$-Cycle} problem can be solved in $1.657^k\operatorname{poly}(n)$ time in general undirected graphs, and in $2^{k/2}\operatorname{poly}(n)$ time in bipartite undirected  graphs, see~Bj\"orklund \emph{et al.}~\cite{BjorklundHKK2017}.

However, in the present paper we study the perhaps even more natural parameterisation of {\HC} that we call {\LC} (it also goes under the slightly misleading name {\textsc{Longest Cycle} in the literature). The problem asks given a graph and a positive integer $k$ whether there is a simple cycle in the graph of length  \emph{at least} $k$; we call such cycles \emph{$k$-long}. We consider the problem in the more general form, which we call {\keLC}, where you are also given an edge~$e$ in the graph, and are asked for a $k$-long cycle that passes through $e$. Note that one can effectively reduce the problem called \textsc{Longest $s,t$-Path} in~\cite{FominGKSS2023_full} (is there a simple path from given vertex~$s$ to given vertex~$t$ on at least $k$~vertices?) to \textsc{$k,st$-Long Cycle}, after possibly adding the  edge $st$ to the graph.
\medskip

Our first result is that {\keLC} can be solved just as fast as the fastest algorithms known for {\keC} in the case of bipartite undirected graphs.

\begin{theorem}[Long cycle in bipartite undirected graphs]
\label{thm: bip}
The problem {\keLC}  in an $n$-vertex bipartite undirected graph can be solved in $2^{k/2}\operatorname{poly}(n)$ time. 
\end{theorem}

Our proof in Section~\ref{sec: bip} is simple and offers insights into how to obtain a speedup over the $2^k\operatorname{poly}(n)$ time bound by Fomin~\emph{et al.}~\cite{FominGKSS2023_full} also in the general graph case.

\begin{theorem}[Main; Long cycle in undirected graphs]
\label{thm: main}
The problem {\keLC} in an $n$-vertex undirected graph can be solved in $1.731^{k}\operatorname{poly}(n)$ time. 
\end{theorem}

This addresses a problem by Fomin \emph{et al.}~(\cite{FominGKSS2023_full}, last paragraph).\footnote{We were unaware that this problem has been solved by Eiben, Koana, and Wahlstr\"om~\cite{Eiben24} with a $1.657^k\operatorname{poly}(n)$ time algorithm using similar but much more general techniques than ours.}

\medskip
At the core of our new {\keLC} algorithm is an algorithm for the problem of finding a cycle with many colours in a coloured graph. This idea is not new. Indeed, it is the very same approach taken by Fomin \emph{et al.}~\cite{FominGKSS2023_full}, and we can use their Theorem~1.4 after minor modifications to obtain the results directly (they use colours on vertices whereas we use colours on edges). However,  while their algorithm itself is relatively simple, their proof of the algorithm's correctness is quite complex. The authors themselves write of their technique as depending on ``one very technical cancellation argument'' whose proof consists of a detailed case analysis spanning many pages. In light of this, we choose to prove our results by distilling the essential problem needed, and supply our own algorithm for it along with a comparatively simple and self-contained proof. The hope is to improve our understanding of why these algorithms work.

To be concrete, we define the {\SkCC} problem. 
You are given an undirected graph $G=(V,E)$, an edge $e\in E$, a target weight~$t$, an edge weight function $w\colon E\rightarrow \{0,1\}$, and a mapping $c\colon E\rightarrow \mathbb{N}$
associating a \emph{colour} with each edge.\footnote{This is an arbitrary labelling, rather than a proper vertex colouring in the sense of chromatic graph theory.}
We say that a cycle $C$ in $G$ is \emph{$k$-colourful of weight~$t$} iff there is a \emph{rainbow} subset $R(C)$ of size $k$ of the edges on $C$ such that all $c(uv)$ for $uv\in R(C)$ are unique, and the total weight $\sum_{uv\in R(C)} w(uv)$ is $t$.
The task is to compute the length of a shortest simple $k$-colourful cycle of weight~$t$ passing through~$e$.

The {\SkCC} problem unifies several well-studied graph problems.
For instance, given a set $\{e_1,\ldots, e_k\}\subseteq E$ of \emph{specified edges}, the mapping $c(e_i)= i$ for those edges (and $c(e)=k+1$ for all other edges) expresses the problem of finding a shortest simple cycle through the specified edges, sometimes called the {\sc Steiner Cycle} or {\sc $T$-Cycle} problem.
Such reductions can be found in \cite{FominGKSS2023_full}, which also traces the rich intellectual history of this problem. 

In the present paper, we use this problem to solve {\keLC}, which naturally corresponds to mapping every edge to a different colour. This is what \cite{FominGKSS2023_full} does for their version of the colourful cycle problem. 
Our main insight behind Theorems~\ref{thm: bip}~and~\ref{thm: main} is that there are more effective reductions from {\keLC} to {\SkpCC}
by an economical use of colours and $k'<k$. We can get some vertices on the sought long cycle almost for free by a more selective edge colouring.

Note that the size of a {\SkCC} need not be bounded by the parameter~$k$, so \emph{a priori} it is not clear that this problem admits an algorithm with running time $f(k)\operatorname{poly}(n)$, or even $n^{O(k)}$.
However, significant progress has been made in recent years, and the following theorem is implicit in the work of Fomin~\emph{et al.}~\cite{FominGKSS2023_full}:

\begin{theorem}[Shortest colourful cycle]
\label{thm: SkCC}
There is a Monte Carlo algorithm for the problem {\SkCC} in an $n$-vertex undirected edge-coloured graph running in $2^k\operatorname{poly}(n)$ time.
\end{theorem}
Fomin~\emph{et al.}~\cite{FominGKSS2023_full} also show a conditional lower bound by arguing that any $(2-\epsilon)^k\operatorname{poly}(n)$ time algorithm for some $\epsilon>0$ would imply an exponentially faster algorithm for the \textsc{Set Cover} problem.

We provide a new algorithm for {\SkCC} expressed as an exponential sum of determinants that runs in $2^k\operatorname{poly}(n)$ time and use only polynomial space in both $n$ and $k$. We analyse it to give a new proof of the above Theorem. In particular we give a relatively short proof of correctness of our algorithm compared to~\cite{FominGKSS2023_full}.

\subsection{Related Work}
Zehavi~\cite{Zehavi2016} showed that {\LC} can be solved in time $t(2k)\operatorname{poly}(n)$ where $t(k')\operatorname{poly}(n)$ is the time to compute a path of length exactly $k'$. In particular, the reduction is deterministic showing that {\LC} is not much harder than {\kC} even for deterministic algorithms. A very simple randomised reduction in the monograph by Cygan~\emph{et al.}~\cite{CyganFKLMPPS2015} (see exercise 5.8) shows this for randomised algorithms. Either there is a cycle of length in the range $[k,2k]$ which we can find by looking for each cycle length explicitly, or we can contract a random edge and repeat. The observation is that the contraction may at most halve the longest cycle length. 

Since \textsc{$k$-Path} and \textsc{$k$-Cycle} can be solved in $2^k\operatorname{poly}(n)$ time in directed graphs, see Williams~\cite{Williams2009}, and in $1.657^k\operatorname{poly}(n)$ time in undirected graphs, see Bj\"orklund \emph{et al.}~\cite{BjorklundHKK2017}, we get a $4^k\operatorname{poly}(n)$ and a $2.746^k\operatorname{poly}(n)$ time algorithm, respectively. 

The two reduction techniques above do not work for the {\keLC} problem though. Fomin~\emph{et al.}~\cite{FominGKSS2023_full} reduced the problem to  a colourful cycle problem, obtaining the running time $2^k\operatorname{poly}(n)$ for the undirected case to show that also {\keLC} is fixed parameter tractable.

In dense graphs, there are better algorithms for {\LC}. Fomin~\emph{et al.}~\cite{FominGSS2022} show that one can in $2^{O(k')}\operatorname{poly}(n)$ time find a cycle of length at least $\delta+k'$ when one exists, where $\delta$ is the average degree in the graph.

\section{The new algorithm for {\SkCC}}

In this Section we describe our new algorithm for {\SkCC} that we will use to prove Theorem~\ref{thm: SkCC}. Our algorithm as its predecessors is algebraic in nature as it computes a carefully designed polynomial in a random point to see if a cycle exists.
With foresight, we will work in a polynomial extension ring of a polynomial-sized field of characteristic~$2$.
Unlike Fomin \emph{et al.}~\cite{FominGKSS2023_full} whose algorithm counts labelled walks by dynamic programming,
our algorithm is an extension of Wahlstr\"om's algorithm for \textsc{$T$-Cycle}~\cite{Wahlstrom2013}, building on sums of determinants. We also use linear algebra ideas tracing back at least to Koutis's work on \textsc{$k$-Path} ~\cite{Koutis2008} to ensure that $k$ different edge colours are used. 

\subsection{Preliminaries and Notation}

\subsubsection{Graphs.}
For a graph $G$ we write $V(G)$ for its vertex set and $E(G)$ for its edge set. 
Undirected edges are denoted by their end-points $u$ and $v$ as $uv$.
We refer to directed edges as `arcs' and use the notation $(u,v)$ with $u$ the arc's initial vertex and $v$ its terminal vertex. 
In a multigraph, when we need to distinguish two arcs with the same endpoints, we sometimes use angled brackets, as in  $\langle u, v\rangle$, for an arc from $u$ to $v$ different from $(u,v)$.

\subsubsection{Finite fields.}

We write $\operatorname{GF}(2^\kappa)$ for the finite field of $2^{\kappa}$ elements, also known as the Galois field and sometimes denoted $\mathbf F_{2^\kappa}$.
This structure exists for every positive integer~$\kappa$ and has $2^{\kappa}$ elements; arithmetic is performed in time polynomial in~$\kappa$, which is negligible for the present paper.
The field $\operatorname{GF}(2^\kappa)$ has \emph{characteristic~$2$}, which means that $e+e = 0$ for every element~$e$.

\subsubsection{Polynomials.}
For a field $\mathbf F$ we write $\mathbf F[\poly X_1,\ldots, \poly X_t]$ for the ring of multivariate polynomials with coefficients from $\mathbf F$.
The \emph{zero polynomial} is the unique polynomial with all coefficients equal to~$0$.
The terms of a polynomial are called monomials.
The (total) degree of a non-zero polynomial is the maximum sum of its monomials' exponents; \emph{e.g.}, the total degree of $\poly X^2\poly Y$ is~$3$.
Our notation distinguishes polynomials and their evaluation;
we denote polynomials by single letters such as $g$ (which is a polynomial in $\mathbf F[\poly X]$) and denote their evaluation at $x$ using an expression with brackets, like $g(x)$ (which is a value in $\mathbf F$).
If $h$ is a monomial in $g$ we write $[h]g$ for the coefficient of $h$ in $g$.

For instance, $g=\poly X^2+\poly X$ is a univariate polynomial of degree~$2$ in $\operatorname{GF}(2)[\poly X]$.
It is not the zero polynomial, even though $g(x)=0$ for all $x\in\operatorname{GF}(2)$.
The coefficient $[\poly X^2]g$ is $1$; coefficient extraction notation extends to multivariate polynomials, so if $g'= \poly X^2\poly Y+\poly X^2$ then $[\poly X^2]g' = \poly Y+1$.

\subsubsection{Matrices and Rank.}
An $r\times k$ matrix $A$ is given in terms of its entries $(a_{ij})$.
We write row vectors as $a\transpose = (a_1,\ldots, a_k)\transpose$ and write the \emph{dot product} of vectors $a$ and $b$ as $a\transpose b = a_1b_1+\cdots a_kb_k$.
A square matrix of $k$ row vectors $a_1\transpose,\ldots, a_k\transpose$ and a vector~$y$ describe a nonhomogenuous systems of $k$ linear equations written as $Ab=y$.
The \emph{rank}~$\rk(A)$ of~$A$ is the number of linearly independent row vectors;
when $\rk(A)=k$ then $A$ is said to have \emph{full rank}.

\subsubsection{Cycle Covers and the Permanent}

Let $D$ be a directed multigraph that may contain loops.
We say that a spanning subgraph $K$ forms a \emph{cycle cover} (less misleadingly a cycle \emph{partition}) if every vertex in $K$ has exactly one outgoing arc and exactly one incoming arc in $E(K)$.
(This can be the same arc in case of a loop.)
A cycle cover can be viewed as a spanning vertex-disjoint union of directed cycles from $D$.

The permanent of an $n\times n$ matrix $M = (m_{ij})$ is 
\begin{equation}\label{eq: per def} 
  \per M = \sum_{\sigma} \prod_{i=1}^n m_{i,\sigma(i)}\,,
\end{equation}
where the sum is over all permutations $\sigma$ on $\{1,\ldots,n\}$.

We need to extend the well-known connection between permanents and cycle covers of \emph{simple} graphs to multigraphs:

\begin{lemma}
\label{lem: cc2per}
  Let $D$ be directed multigraph with loops at every vertex and set $V(D)= \{v_1,\ldots, v_n\}$.
  Assume the arcs $e\in E(D)$ are labelled with values~$\lambda(e)$ from a commutative ring.
  Let $\mathscr C$ denote the family of cycle covers of~$D$.
  Consider the $n\times n$ matrix~$M= (m_{ij})$ whose $ij$th entry is the sum of the labels of arcs from $v_i$ to $v_j$, formally
  \[ m_{ij} = \sum_{\substack{\operatorname{init}(e) = v_i\\\operatorname{ter}(e) = v_j}} \lambda(e)\,.\]
  Then 
  \[ \per M = \sum_{K\in \mathscr C} \prod_{e\in E(K)} \lambda(e)\,.\]
\end{lemma}

The proof is routine and can be found in the Appendix. 
The computational importance of the above lemma is that in a ring of characteristic~$2$, the permanent coincides with the determinant, which has efficient polynomial time algorithms.
 
\subsection{Overview}
We are given as input an undirected, loopless, simple graph $G$ on $n$ vertices and $m$ edges, a specified edge~$v_1v_2\in E(G)$, an edge colouring $c\colon E(G)\rightarrow \mathbb{N}$, and an edge weighting $w\colon E(G)\rightarrow \{0,1\}$. 
Let $s$ be the size of the image of $c$, and remap $c$ if needed to the integers $1,2,\cdots, s$.
The reader loses little by setting $w(e)=0$ at first reading.

With foresight, we consider the field $\mathbf{F}=\operatorname{GF}(2^{1+\lceil\log_2 n\rceil})$
and the ring of polynomials 
\[
  \mathbf G =
  \mathbf F[\{\poly R_{uv}\}, \{\poly X_{uv}\}, \poly W, \poly Y, \poly Z]
\]
of $2m+2$ indeterminates.
There is an indeterminate $\poly R_{uv}$ for every $uv\in E(G)$,
an indeterminate $\poly X_{uv}$ for every $uv\in E(G)$ except $v_1v_2$,
and $3$~more indeterminates $\poly W$,  $\poly Y$ and $\poly Z$.

Let $A$ be an $s\times k$ matrix consisting of $\{0,1\}$-valued row vectors:
\[ A = \begin{bmatrix}
  a_1\transpose\\\vdots\\a_s\transpose
\end{bmatrix},\quad \text{where } a_i\in \{0,1\}^k\,.\]
With foresight, the role of the \emph{filter bit vector} $b\in \{0,1\}^k$ is to orient the $k$ chosen arcs of different colours in all $2^k$ possible ways, which will make sure that only a single cycle containing all of them will contribute to the sum (this mimics the idea for \textsc{$T$-Cycle} by Wahlstr\"om~\cite{Wahlstrom2013}).
The expression $\lambda_{b,A}(K)$ associates with each cycle cover $K\in \mathscr C$ a monomial from $\mathbf G$.
We then define the polynomial $g\in \mathbf G$ of total degree $n$
\begin{equation}
\label{eq: def g}
  g = \sum_{K \in \mathscr C} \sum_{b\in \{0,1\}^k} \lambda_{b,A}(K)
\end{equation}
where $\mathscr C$ is the family of cycle covers in a related directed multigraph $G'$, described in the following subsection.
With large probability for random~$A$, the polynomial $g$ will be non-zero when $G$ contains a $k$-colourful cycle through~$v_1v_2$.

\medskip
We begin by describing $G'$ and $\lambda_{b,A}$.

\subsection{A Directed Graph Labelled With Indeterminates}
\label{sec: lambda def}

From $G$ and the fixed edge $v_1v_2\in E(G)$ we define a labelled multigraph $G'$ with self-loops as follows.
The vertices of $G'$ are the same as $G$.
We impose an arbitrary total order on the vertices, so it makes sense to write $u<v$ for $u,v \in V$.

Each arc $e\in E(G')$ belongs to one of three sets (called self-loops, rainbow arcs, and extra edges) and is labelled with a multilinear monomial $\lambda_{b,A}(e)\in \mathbf G$ that may depend on $b\in\{0,1\}^k$ and $A\in \{0,1\}^{s\times k}$:
\begin{description}
  \item[Self-loops] 
    There are $n-2$ self-loops $(u,u)$, one for every $u\in V$ except $v_1$ and $v_2$.
    Each self-loop is labelled by the same indeterminate $\lambda_{b,A}((u,u)) = \poly Z$.
  \item[Rainbow arcs] 
    There is a set $R(G')$ of $2m-1$ directed \emph{rainbow} arcs $(u,v)$ and $(v,u)$, one pair for each $uv\in E(G)$, except for $(v_2,v_1)$.
    The label of rainbow arc~$(u,v)$ contains a \emph{rainbow} indeterminate $\poly R_{uv}$ and depends on the arc's orientation, the row of $A$ corresponding to the colour $c(uv)$ of $uv$ in $G$, and the vector~$b$:
    Rainbow arcs corresponding to edges of weight~$1$ also contribute the \emph{weight} intermediate~$\poly W$.
    \begin{equation}\label{eq: rainbow label}
      \lambda_{b,A}((u,v)) = \begin{cases}
	(a_{c(uv)}\transpose b) \poly R_{uv}\poly W^{w(uv)}\,, & \text{if $u < v$}\,;\\
	(1 + a_{c(uv)}\transpose b) \poly R_{uv} \poly W^{w(uv)}\,, & \text{if $u > v$}\,.
      \end{cases}
    \end{equation}
    Recall that computation is in characteristic~$2$, so $a_{c(uv)}\transpose b$ is either $0$ or~$1$.
    Also note that the antiparallel pair $(u,v)$ and $(v,u)$ of rainbow arcs use the same indeterminate, and that exactly one of $\lambda_{b,A}((u,v))$ and $\lambda_{b,A}((v,u))$  is $\poly R_{uv}\poly W^{w(uv)}$ (the other is $0$).
  \item[Extra arcs] 
    There is a set $X(G')$ of $2(m-1)$ \emph{extra} arcs $\langle u,v\rangle$ and $\langle v,u\rangle$, one pair for each $uv\in E(G)\setminus\{_1v_2\}$.
    The label~$\lambda_{b,A}(\langle u,v\rangle)$ is the bivariate monomial $\poly X_{uv}\poly Y$ and does not depend on $b$ or~$A$.
    Note that (unlike rainbow arcs), the antiparallel extra arcs $\langle u,v\rangle$ and $\langle v,u\rangle$ have the same label.
\end{description}
For some minimal examples, if $G$  is
\(
\vcenter{
  \hbox{
    \begin{tikzpicture}[scale=.8]
    \node (1) [circle, draw, inner sep=1pt] { \sf\small 1};
    \node (2) at (1,0) [circle, draw, inner sep=1pt] { \sf\small 2};
    \node (3) at (2,0) [circle, draw, inner sep=1pt] { \sf\small 3};
    \draw (1) -- (2) -- (3);
  \end{tikzpicture}
  }
}
\)
coloured as $c(12)=1$, $c(23)=2$, weighted as $w(12)=0$, $w(23)=1$,
and $A=\begin{pmatrix} 1&0\\0&1\end{pmatrix}$, then we have the labelings
  \[
    \vcenter
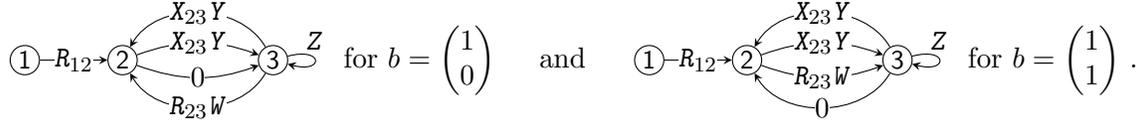
{\hbox{
    \begin{tikzpicture}
      \node (1) [circle, draw, inner sep=1pt] {\sf \small 1};
      \node (2) at (1.3,0) [circle, draw, inner sep=1pt] {\sf \small 2};
      \node (3) at (3.3,0) [circle, draw, inner sep=1pt] {\sf \small 3};
      \draw [-stealth] (1) to node [inner sep=0pt, fill=white] {$\poly R_{12}$} (2);
      \draw [-stealth] (2) [bend left=20] to node [inner sep=0pt, fill=white] {$\poly X_{23}\poly Y$} (3);
      \draw [-stealth] (3) [bend right=60] to node [inner sep=0pt, fill=white] {$\poly X_{23}\poly Y$} (2);
      \draw [-stealth] (2) [bend left=-20] to node [inner sep=0pt, fill=white] {$0$} (3);
      \draw [-stealth] (3) [bend right=-60] to node [inner sep=0pt, fill=white] {$\poly R_{23}\poly W$} (2);
      \draw [-stealth] (3) [loop right] to node [above] {$\poly Z$} (3);
    \end{tikzpicture}}}
    \text{ for } b=\begin{pmatrix}1\\0\end{pmatrix}
      \quad
    \text{ and }
      \quad
    \vcenter{\hbox{
    \begin{tikzpicture}
      \node (1) [circle, draw, inner sep=1pt] {\sf \small 1};
      \node (2) at (1.3,0) [circle, draw, inner sep=1pt] {\sf \small 2};
      \node (3) at (3.3,0) [circle, draw, inner sep=1pt] {\sf \small 3};
      \draw [-stealth] (1) to node [inner sep=0pt, fill=white] {$\poly R_{12}$} (2);
      \draw [-stealth] (2) [bend left=20] to node [inner sep=0pt, fill=white] {$\poly X_{23}\poly Y$} (3);
      \draw [-stealth] (3) [bend right=60] to node [inner sep=0pt, fill=white] {$\poly X_{23}\poly Y$} (2);
      \draw [-stealth] (2) [bend left=-20] to node [inner sep=0pt, fill=white] {$\poly R_{23}\poly W$} (3);
      \draw [-stealth] (3) [bend right=-60] to node [inner sep=0pt, fill=white] {$0$} (2);
      \draw [-stealth] (3) [loop right] to node [above] {$\poly Z$} (3);
    \end{tikzpicture}}}
    \text{ for } b=\begin{pmatrix}1\\1\end{pmatrix}\,.
  \]

We extend the labels from arcs to arbitrary subgraphs $H\subseteq G'$ in the natural fashion:
\begin{equation}\label{eq: lambda def}
  \lambda_{b,A}(H) = \prod_{e\in E(H)} \lambda_{b,A}(e)\,.
\end{equation}
\subsection{A Matrix of Weighted Polynomials}

Consider the $n\times n$ matrix $M_b= (m_{u,v})$ whose entries are multilinear polynomials $m_{u,v}\in\mathbf G$ given by 
\begin{equation}\label{eq: M def}
  m_{u,v} = \begin{cases}
    \poly X_{uv}\poly Y + (a_{c(uv)}\transpose b) \poly  R_{uv} \poly W^{w(uv)} \,, & u<v, uv\in E(G)\setminus v_1v_2\,;\\
    \poly X_{uv}\poly Y + (1+a_{c(uv)}\transpose b) \poly  R_{uv}  \poly W^{w(uv)}\,, & u>v, uv\in E(G)\setminus v_1v_2\,;\\
    (a_{c(uv)}\transpose b) \poly R_{uv} \poly W^{w(uv)}\,, & (u,v) = (v_1,v_2);\\
    \poly Z\,, & u=v \,;\\
    0\,, & (u,v) = (v_2, v_1) \text{ or } uv\notin E(G)\,.\\
  \end{cases}
\end{equation}

This allows an alternative expression for $g$, in terms of the permanent:

\begin{lemma}\label{lem: h}
  \[
   \sum_{b\in \{0,1\}^k} \per M_b =  g\,.
 \]
\end{lemma}

\begin{proof}
  Observe that for each $u,v$, the entry $m_{u,v}$ is exactly the sum of the labels $\lambda_{b,A}(e)$ of all arcs $e$ from $u$ to $v$ as defined in Section~\ref{sec: lambda def}.
  Thus, Lemma~\ref{lem: cc2per} gives
  \[ 
  \sum_{b\in\{0,1\}^k} \per M_b = 
  \sum_{b\in \{0,1\}^k} \sum_{K\in\mathscr C} \prod_{e\in C} \lambda_{b,A}(e)  = 
  \sum_{K\in\mathscr C}
  \sum_{b\in \{0,1\}^k}  \lambda_{b,A}(K) = 
 g\,.\qedhere
  \]
\end{proof}

\subsection{Algorithm for {\SkCC}}
\label{sec: algorithm}

We are now ready to state our algorithm for \SkCC.
Put very succinctly, with $m_{u,v}=m_{u,v}^{(b,A)}$ we just need to compute the smallest $i$ such that
\[
  [\poly W^t \poly Y^{i-k}\poly Z^{n-i}] \sum_{b\in \{0,1\}^k}\per \bigl((m_{u,v}^{(b,A)})(\{r_{uv}\}, \{x_{uv}\})\bigr) \neq 0
\]
at random $\{r_{uv}\}$, $\{x_{uv}\}$ and $A\in\{0,1\}^{s\times k}$.
Here is a more careful explanation of the procedure:

\medskip\noindent
{\bf Algorithm A}. 
The input is a graph $G$, an edge colouring $c:E(G)\rightarrow \{1,\ldots, s\}$, an edge weight function $w:E(G)\rightarrow \{0,1\}$, with specified edge $v_1v_2$, and integers $k$ and $t$.
This algorithm returns the smallest $l$ for which $G$ contains a $k, v_1v_2$-colourful cycle of length~$l$ and weight~$t$.

\begin{description}
  \item[A1] (\emph{Pick random $A$.}) Generate the $s\times k$ matrix $A=(a_{ij})$ by picking $a_{ij}\in\{0,1\}$ uniformly and independently at random.
  \item[A2] (\emph{Pick random sieving weights.}) Pick $2m-1$ values $r_{uv}, x_{uv}\in \mathbf F$ uniformly and independently at random.
  \item[A3] (\emph{Evaluate entries of $M_b$.})
    For each $b\in\{0,1\}^k$, construct matrix $M_b'$ whose $(u,v)$th entry is the bivariate polynomial $m_{u,v}(\{r_{uv}\}, \{x_{uv}\})\in \mathbf F [\poly W, \poly Y, \poly Z]$ as defined in Eq.~\eqref{eq: M def}, where the indeterminates $\poly R_{uv}$ and $\poly X_{uv}$ are replaced by the values $r_{uv}$ and  $x_{uv}$, respectively.
  \item[A4] (\emph{Compute sum of permanents.}) 
    Compute the polynomial $h\in \mathbf F[\poly W, \poly Y, \poly Z]$ as \[h=\sum_{b\in\{0,1\}^k} \per M_b'\,.\]
  \item[A5] 
    (\emph{Determine length of shortest contributing cycle.})
    If $[\poly W^t]h$ is the zero polynomial, return $\infty$.
    Otherwise return the smallest $i$ for which the coefficient $[\poly W^t\poly Y^{i-k}\poly Z^{n-i}]h$ is not zero.
\end{description}

Think of $h$ as $g(\{r_{uv}\}, \{x_{uv}\})$, \emph{i.e.}, an evaluation of $g$ at random sieving weights.
Formally, the entries of matrix $M_b'$ are $m_{u,v}(\{r_{uv}\}, \{x_{uv}\})$, so Lemma~\ref{lem: h} gives
\begin{equation}\label{lem: h as g}
  h = 
  \sum_{b\in\{0,1\}^k} \per M_b' = 
  \left(\sum_{b\in\{0,1\}^k} \per M_b\right)(\{r_{uv}\}, \{x_{uv}\}) = 
  g(\{r_{uv}\}, \{x_{uv}\})\,.
\end{equation}
Algebraically, it does not matter whether the substitution of values $r_{uv}$ and $x_{uv}$ for the indeterminates $\poly R_{uv}$ and $\poly X_{uv}$ happens before or after the permanent is computed, but computationally it makes all the difference because the polynomials in $\mathbf G$ can have exponentially many terms.

\begin{lemma}[Running time]
  Algorithm {\bf A} runs in time $O(2^k\operatorname{poly}(n))$.
\end{lemma}
\begin{proof}
  The polynomials in  $\mathbf F[\poly W,\poly Y, \poly Z]$ are represented by their coefficients;
  the maximum total degree is~$2n$. 
  In characteristic~$2$, the permanent equals the determinant, so in step {\bf A3} we can use
  any division free polynomial time determinant algorithm that works over a commutative ring with unity, \emph{e.g.}, the one by Berkowitz~\cite{Berkowitz1984}.
  There are $2^k$ such computations.
\end{proof}

The correctness proof requires a thorough investigation of the properties of $g$ and takes up the rest of this section.

\subsection{Properties of Labelled Cycle Covers}

\subsubsection{Label Sums}

We will now consider the total contribution
\[\Lambda_A(H) = \sum_{b\in\{0,1\}^k} \lambda_{b,A} (H) \,\]
of a subgraph $H$ over all filter bit vectors $b$, so that from~\eqref{eq: def g} we have $g=\sum_{K\in\mathscr C} \Lambda_A(K)$.

It will be useful to introduce the $r\times k$ matrix $A_H$ formed by those rows $a_{c(uv)}\transpose$ of $A$ that correspond to the colours on the $r$ rainbow arcs $R(H)$ of $H$.
To be precise, let $R(H)=\{e_1,\ldots, e_r\}$ be an arbitrary enumeration of the rainbow arcs of $H$,
for definiteness we let the enumeration agree with the ordering on the initial vertices of each arc.
Then we let
\begin{equation}
\label{eq: A_H}
  A_H = \begin{pmatrix}
  a_{c(e_1)}\transpose \\
  \vdots\\
  a_{c(e_r)}\transpose
\end{pmatrix}
  \,.
\end{equation}
Note that $A_H$ can have repeated rows.

We can give a concise expression for the contribution $\Lambda_A(H)$ of subgraphs with $k$~rainbow arcs:

\begin{lemma}[Label sum for subgraphs]
\label{lem: labelsum}
  Let $H$ be a subgraph consisting of $k$ rainbow arcs $R(H)$ of total weight~$t$, as well as $i-k$ extra arcs $X(H)$ and $n-i$ self-loops.
  Then 
  \begin{equation}
  \label{eq: labelsum}
    \Lambda_A(H) =
    2^{k - \rk A_H}
  \left(\prod_{(u,v)\in R(H)} \poly R_{uv} \right)
  \left(\prod_{\langle u, v\rangle\in X(H)} \poly X_{uv} \right)
    \poly W^t
    \poly Y^{i - k} 
    \poly Z^{n-i} \,.
\end{equation}
\end{lemma}

In particular, if $r=k$ and $A_H$ has full rank, this is a nonzero monomial.
Otherwise, since $\mathbf F$ has characteristic~$2$, the expression vanishes.

\begin{proof}
  Returning to \eqref{eq: lambda def} we have
  \begin{equation}\label{eq: lambda_x(H)}
    \lambda_{b,A}(H) =  
    \left(
    \prod_{(u,v)\in R(H)} \poly \lambda_{b,A}((u,v)) 
    \right)
    \left( \prod_{\langle u, v\rangle\in X(H)} \poly X_{uv}\right)
    \poly Y^{i-k} 
    \poly  Z^{n-i}
    \,.
  \end{equation}
  Thus it suffices to understand the first product.
  From the definition~\eqref{eq: rainbow label} we see that this expression either is $0$ (namely, when at least one of the coefficients to some $\poly R_{uv}$ is $0$) or the $k$-variate nonzero monomial $\prod_{(u,v)\in R(H)} \poly R_{uv}$ (namely, when the coefficient of every $\poly R_{uv}$ is $1$).
  Thus, since we are computing in characteristic~$2$, the only way for $\lambda_{b,A}(R(H))$ to \emph{not} vanish is if
  \[ a_{c(uv)}\transpose b  =
  \begin{cases}
    1\,, & \text{if $u<v$}\,;\\
    0\,, & \text{if $u>v$}\,,\\
   \end{cases}
   \qquad\text{for each rainbow arc $(u,v) \in R(H)$}\,.\]
  The above condition is a set of $k$ constraints on the filter bit vector~$b$, each using a row of $A_H$, so they can be succinctly expressed in matrix form as
  \[
    (A_H) b = y\,,
  \]
  where $y\in\{0,1\}^r$ is defined by $y_i=1$ if $u<v$ and $y_i=0$ otherwise.
  We conclude that
  \[ \lambda_{b, A}(R(H)) = 
  \begin{cases}
    \prod_{(u,v)\in R(H)} \poly R_{uv}\poly W^{w(uv)}\,, & \text{if } (A_H) b= y\,;\\
    0\,, & \text{otherwise} \,.
  \end{cases}\]
  The number of $b$ for which $(A_H)b =y$ is $2^{k-\rk A_H} \pmod{2}$.
  Thus, we have
  \[ \sum_{b\in\{0,1\}^k}
  \lambda_{b, A}(R(H)) =
  2^{k-\rk A_H} \! \prod_{(u,v)\in R(H)} \poly R_{uv}\poly W^{w(uv)} = 
  2^{k-\rk A_H} \poly W^t\! \prod_{(u,v)\in R(H)} \poly R_{uv}\,.
  \qedhere
  \]
\end{proof}

We infer that insufficiently colourful subgraphs do not contribute to $g$:

\begin{lemma}[Subgraphs must be multi-coloured]
  \label{lem: at least k}
  Let $H$ have $k$ rainbow arcs.
  If the rainbow arcs have fewer than $k$ distinct colours, 
  then for all $A$, we have $\Lambda_A (H) = 0$.
\end{lemma}

\begin{proof}
  The matrix~$A_H$ defined in Eq.~\eqref{eq: A_H} has fewer than $k$ different rows,
  so $\rk A_H < k$.
  Thus, $2^{k-\rk A_H}$ is even and $\Lambda_A(K)$ as seen in Equation~\eqref{eq: labelsum} vanishes in characteristic~$2$.
\end{proof}

\subsubsection{Stratified Cycle Covers}

It will be convenient to group the cycle covers $\mathscr C$ according to how many arcs of each type they contain.
For $r\in\{0,\ldots,n\}$, $s\in\{n-r, \ldots n\}$, and $q\in\{0,\ldots, r\}$ we let $\mathscr C_{q, r,s}$ denote the cycle covers containing $r$~rainbow arcs of total weight $q$, and $s$~self-loops (and therefore $n-r-s$~extra arcs.)
We can then collect terms 
\[
  g = \sum_{K\in\mathscr C} \Lambda_A(K) = 
  \sum_{r=0}^n \sum_{s=n-r}^n\sum_{q=0}^r \poly W^q \poly Y^{n-r-s} \poly Z^s
  f_{q, r,s}\,,
\]
where we define the polynomials $f_{q,r,s}\in\mathbf F[\{\poly R_{uv}\},\{\poly X_{uv}\}]$ as
\[ 
f_{q,r,s} = \sum_{K\in \mathscr C_{q,r,s}}[\poly W^q \poly Y^{n-r-s} \poly Z^s] \Lambda_A(K)
\,.
\]
Of particular interest are the covers containing exactly $k$~rainbow edges of total weight~$t$ and $n-i$ self loops (and therefore $i-k$ extra arcs), so we define $\mathscr C_i = \mathscr C_{j, k, n-i}$ and the \emph{length-indexed polynomials}~$f_i = f_{j, k, n-i}$ for $i\in\{k,\ldots, n\}$.
Using Lemma~\ref{lem: labelsum}, we can express these polynomials as
\[
  f_i = 
  \sum_{K\in \mathscr C_i}
  2^{k-\rk A_K} 
  \left(\prod_{(u,v)\in R(K)} R_{uv}\right)
  \left(\prod_{\langle u,v\rangle\in X(K)} X_{uv}\right)\,.
\]
We proceed to show that $f_k$, $\ldots$, $f_{l-1}$ are the zero polynomial but $f_l$ (probably) is not.

\subsubsection{Defective Cycle Covers.}
\label{sec: defective}

A cycle $D$ is \emph{defective} if it contains a rainbow arc but not $(v_1,v_2)$.
Let $\mathscr D_i\subseteq\mathscr C_i$ denote the cycle covers that contain a defective cycle, and call them defective cycle covers.

\begin{lemma}[Defective cycle covers vanish]
  \label{lem: unicycle}
  For every $A$,
  \[
   \sum_{K\in\mathscr D_i} \Lambda_A(K) = 0\,. 
 \]
\end{lemma}
\begin{proof}
  For cycle cover $K\in\mathscr D_j$ let $D=D(K)$ be its defective cycle containing the smallest vertex.

  We will pair up the covers  $K\in\mathscr D_j$ with $\Lambda_A(K)\neq 0$ by reversing the defective cycle's orientation.
  To be precise, the cover~$K$ is paired with the cover~$K'\in\mathscr D_j$ where $D(K')$ is $D$ reversed (each rainbow arc $(u,v)\in E(D)$ becomes $(v,u)$, each extra arc $\langle u,v\rangle$ becomes $\langle v, u\rangle$) and all other cycles are the same.
  We need to check that this is indeed a pairing (more formally, that it defines a fixed-point free involution on the remaining covers of $\mathscr D_j$.)
  Indeed, repeating the reversal brings $K$ back to itself (so $K''=K$) and therefore the pairing is unique.
  Moreover, $K\neq K'$, because $K$ contains a rainbow arc, say $(u,v)$, and $K'$ its reversal $(v,u)$.
  Yet $K$ cannot itself include $(v,u)$, because otherwise $\Lambda_A(K)=0$ since either of the terms $\lambda_{b,A}((u,v))$ or $\lambda_{b,A}((v,u))$ is $0$, see~\eqref{eq: lambda def}.

  The argument finishes by observing that paired covers $K$ and $K'$ have the same label.
  To see this, the rainbow arcs of $K'$ are those of $K$, only with some of them reversed, so they have the same colour.
  Thus, the matrices $A_K$ and $A_{K'}$ as defined by \eqref{eq: A_H} are the same and have full rank. There is precisely one $b$-vector satisfying
  $A_{K'}b=y_{K'}$,  and the products $\prod_{(u,v)\in R(K)}\poly R_{uv}$ and $\prod_{(u,v)\in R(K')}\poly R_{uv}$ are equal.
  Thus, by Lemma~\ref{lem: labelsum}, we have $\Lambda_A(K) + \Lambda_A(K')=0$ in characteristic~$2$ and the claim follows.
\end{proof}

\begin{lemma}\label{lem: f_i zero}
  Let $l$ be the minimal length of any $k, v_1v_2$-colourful cycle in $G$.
  Then $f_i=0$ for $k\leq i<l$.
\end{lemma}

\begin{proof}
  Assume $K\in \mathscr C_i$ contributes to $f_i$.
  It contains $k$ rainbow arcs and $i-k$ extra arcs.
  Lemma~\ref{lem: at least k} says that the rainbow arcs all have distinct colours, and
  Lemma~\ref{lem: unicycle} says that these rainbow arcs
  all lie on a common cycle~$C$ in~$G'$, which also includes the arc~$(v_1,v_2)$.
  Because $C$ describes a $k$,~$v_1v_2$-colourful cycle in~$G$, its length $|C|$ is at least~$l$.
  On the other hand, $K$ contains only $i$~arcs that aren't self-loops, so $|C| \leq i$ and therefore $i \geq l$.
\end{proof}

\subsubsection{The Canonical Cycle Cover}

Let $C$ be a solution to the  {\SkCC} problem with rainbow subset~$R(C)$.
We construct the \empty{canonical cycle cover} $K_C$ in $G'$ as follows.
Enumerate the vertices of $C$ as $v_1,v_2,\ldots, v_l$ with $v_1= v_l$, so that $v_iv_{i+1}\in E(G)$.
This sequence naturally describes a directed cycle $C'$ in $G'$ starting with $(v_1,v_2)$;
if $v_iv_{i+1}\in R(C)$ we pick the rainbow arc $(v_i,v_{i+1})\in R(G')$, else we pick the extra arc $\langle v_i,v_{i+1}\rangle \in X(G')$.
The resulting directed cycle $C'$ contains exactly $k$ rainbow arcs, each corresponding to a different colour on~$C$.
To complete the cover $K_C$, for each of the $n-l$ remaining vertices $v\notin V(C)$ we add the self-loop $(v,v)$ to $K_C$.

\begin{lemma}
  Assume $G$ contains a $k, v_1v_2$-colourful cycle $C$ of length~$l$.
  Then the polynomial $f_l$ contains the monomial
  \begin{equation}\label{eq: lambda canonical}
    2^{k-\rk A_{K_C}}
    \left(\prod_{\langle u, v\rangle\in X(K_C)} \poly X_{uv} \right)
    \left(\prod_{(u,v)\in R(K_C)} \poly R_{uv} \right)\,,
  \end{equation}
\end{lemma}
\begin{proof}
  Clearly, $K_C$ belongs to $\mathscr C_l$, and
  according to  Lemma~\ref{lem: labelsum}, the contribution of $\Lambda_A(K_C)$ is exactly given by~\eqref{eq: lambda canonical}.

  We need to check that no other cover contributes the same monomial (which might cancel $\Lambda_A(K_C)$).
  Thus, let $K\in \mathscr C_l$ be a cycle cover contributing \eqref{eq: lambda canonical}.
  Except for their orientation, we can infer the arcs of $K$ from the indeterminates, so $K$ must be an orientation of the arcs of~$C$.
  Since $K$ is a cycle cover, these orientations must agree to form a directed cycle.
  Finally, the orientation of this cycle is fixed because we constructed $G'$ to contain an arc from  $v_1$ to $v_2$ but not vice versa.
  We conclude that $K$ equals $K_C$.
\end{proof}

We now apply Koutis's~\cite{Koutis2008} observation and proof that a random set of $k$~vectors from $\{0,1\}^k$ are linearly independent with constant non-zero probability.

\begin{lemma}\label{lem: random matrix rank}
  Let $B$ be a random $k\times k$ matrix with values from $\{0,1\}$.
  Then
  \[
    \Pr(\rk B = k)> \tfrac{1}{4}\,.\] 
\end{lemma}
\begin{proof}
  The number of $k\times k$-matrices from $\{0,1\}$ with full rank is \[
    (2^k-2^0) (2^k -2^1)\cdots (2^k-2^{k-1}) =\prod_{i=0}^{k-1} (2^k - 2^i)
    =
    2^{k^2}\prod_{i=0}^{k-1} \biggl(1-\frac{1}{2^{k-i}}\biggr)
    =
    2^{k^2}\prod_{i=1}^{k} \biggl(1-\frac{1}{2^i}\biggr)
    \,.\]
  The random matrix $B$ is sampled using $k^2$ choices from $\{0,1\}$, so we have
  \[
   \Pr(\rk(B)= k) = 
  \prod_{i=1}^k \biggl(1-\frac{1}{2^i}\biggr) \geq
  \tfrac{1}{2}\prod_{i=2}^k \biggl(1-\frac{1}{i^2}\biggr) = 
  \tfrac{1}{2}\prod_{i=2}^k \frac{(i+1)(i-1)}{i^2} =
  \frac{k+1}{4k} > \tfrac{1}{4}\,. 
  \qedhere
\]
\end{proof}

We conclude that unlike $f_k,\ldots, f_{l-1}$, the polynomial $f_l$ (probably) does not vanish:
\begin{lemma}\label{lem: f_l nonzero}
If  $A$ is chosen uniformly at random from $\{0, 1\}^{s\times k}$ then
  \(\Pr(f_l\neq 0) > \tfrac14\).
\end{lemma}
\begin{proof}
  The canonical cover $K_C$ contains $k$ rainbow edges, so $A_{K_C}$ is a $k\times k$ submatrix of $A$.
  Thus, from the above lemma, with probability more than $\frac14$, the coefficient $2^{k-\rk A_{K_C}}$ in~\eqref{eq: lambda canonical} is~$1$.
\end{proof}

\subsection{Correctness.}
We are ready to establish correctness of Algorithm~{\bf A}.
We need the famous lemma stating that randomised polynomial identity testing is efficient:

\begin{lemma}[DeMillo--Lipton--Schwartz--Zippel~\cite{DeMilloL1978,Schwartz1980,Zippel1979}]
\label{lem: DLSZ}
Let $f$ be a non-zero polynomial in $m'$ variables of total degree~$d$ over a field $\mathbf{F}$ with $d<|\mathbf{F}|$. 
  Then for uniformly chosen random points $x_1,\ldots, x_{m'}\in \mathbf F$ we have
\[
  \Pr\bigl(f(x_1,\ldots, x_{m'})\neq 0\bigr)\geq 1-\frac{d}{|\mathbf{F}|}.
\]
\end{lemma}

\begin{lemma}
  Algorithm {\bf A} returns the length of a shortest $k,v_1v_2$-colourful cycle of weight~$t$ with probability $\frac18$.
  If no such cycle exists, it returns $\infty$.
\end{lemma}

\begin{proof}
  Assume $G$ contains a shortest $k$-colourful cycle $C$ of length $l$. 
  For each $i\in \{k, \ldots, n\}$, the values computed in {\bf A5} can be expressed as 
  \[ 
  [\poly W^t\poly Y^{i-k} \poly Z^{n-i}] h =
  [\poly W^t\poly Y^{i-k} \poly Z^{n-i}] g (\{r_{uv}\}, \{x_{uv}\}) = 
  f_i (\{r_{uv}\}, \{x_{uv}\}) \,.\]
  By Lemma~\ref{lem:  f_i zero}, all of $f_k,\ldots, f_{l-1}$ are the zero polynomial, so none of them can evaluate to nonzero.
  From Lemma~\ref{lem: f_l nonzero}, we have 
  \( \Pr(f_l\neq 0) > \tfrac 14\),
  and from Lemma~\ref{lem: DLSZ} with $d\leq l\leq n$, we have
  \[
    \Pr ( f_l (\{r_{uv}\}, \{x_{uv}\}) \neq 0 \mid f_l\neq 0) \geq 1-\frac{n}{2^{1 +\lceil \log_2 n \rceil}}\geq \tfrac 12\,.
  \]
  Hence, the algorithm correctly returns~$l$ in {\bf A5} with probability at least $\frac 14\cdot\frac 12
\geq \tfrac{1}{8}$.
\end{proof} 

This concludes our proof of Theorem~\ref{thm: SkCC}.

\section{New Reductions from $k,e$-\textsc{Long Cycle}}
In this section we prove Theorem~\ref{thm: bip}~and~\ref{thm: main}. Remember that in the parameterised {\keLC} problem, you are given a positive integer $k$, a graph $G$, and an edge $e$ in the graph, and are tasked to decide whether there is a cycle of length at least $k$ through $e$ in the graph $G$. 

\subsection{The bipartite undirected case}
\label{sec: bip}
We first address Theorem~\ref{thm: bip}. 
In a bipartite undirected  graph $G$, we can reduce \textsc{$k,e$-Long Cycle} to \textsc{Shortest $\tfrac{k}{2},e$-Colourful Cycle} as follows.
Let $G=(U,V,E)$ be the given bipartite graph, with the fixed edge $e\in E$, and identify $V$ with the integers $\{1,2,\cdots,|V|\}$. As input to \textsc{Shortest $\tfrac{k}{2},e$-Colourful Cycle} we simply take $G$ and $e$ with the edge colouring $c(uv)=v$ for $u\in U,v\in V,uv\in E$. If a cycle passes through an edge of colour $v$, it must also necessarily pass through the vertex $v$. Hence we note that a cycle is $\tfrac{k}{2}$-colourful if and only if it passes through at least $\tfrac{k}{2}$ vertices in $V$. Since every cycle alternates between vertices in $U$ and in $V$, we get that a detected $\tfrac{k}{2}$-colourful cycle in $G$ has length at least $k$ in $G$. From Theorem~\ref{thm: SkCC} we see that our resulting running time is $2^{k/2}\operatorname{poly}(n)$.

\subsection{The general undirected case}
We next turn to Theorem~\ref{thm: main}. We begin by describing the algorithm. We will combine two algorithmic components to find a $k$-long cycle through $e=v_1v_2$ when it exists. Identify as before the vertices $V(G)$ with the integers $\{1,2,\cdots,n\}$.

Our first component uses the $k',e$-\textsc{Cycle} algorithm by Bj\"orklund \emph{et al.}~\cite{BjorklundHKK2017} for each $k'=k,k+1,\cdots \min(n,\lfloor \beta k \rfloor)$, to see if there is a cycle of length between $k$ and $\beta k$ for a parameter $\beta>1$ to be fixed later. The running time of this part is at most $1.657^{\beta k}\operatorname{poly}(n)$.

Our second component first samples a vertex subset $S\subseteq V(G)$ by including each vertex in $S$ independently at random with probability $\alpha$. 

We will assign a set of available colours $c:E(G)\rightarrow 2^{[s]}$ to each edge in $G$. To turn this into an input to \textsc{Shortest Colourful Cycle},
we can first intermediately think of $G$ as a multigraph with parallel edges with a single colour $c'$ for each colour $c'\in c(uv)$. This multigraph can next be transformed into a simple graph $G_S$ with edge colouring $c_S:E(G_S)\rightarrow [s]$ by subdividing every edge with both created edges retaining the colour of the original edge. 
Note that any cycle in $G_S$ of length at least $\ell>4$ corresponds to a cycle in $G$ of length $\ell/2$.

We set $c(uv)$ for $uv\in E(G)$ to $\{u\}$ if $u\in S$ and $v\not \in S$, and to $\{v\}$ if $v\in S$ and $u\not \in S$. If both $u,v\in S$, we set $c(uv)=\{u,v\}$.
If neither $u,v\in S$, we let $c(uv)$ be a set containing a single colour $>n$, that is unique for that edge in $G$.
We say an edge $uv\in E(G)$ is \emph{split} if exactly one of $u$ and $v$ are in $S$. 
Split edges $uv\in E(G)$ are given the weight $w(uv)=1$ whereas other edges $uv$ get weight $w(uv)=0$. Finally, we solve \textsc{Shortest $k',t, e$-Colourful Cycle} on $G_S,e,c_S$ with $k'=\lceil (1-(1-\epsilon)(1-\alpha)(\alpha+\alpha^2)\beta/2)k \rceil$ and target weight $t \geq (1-\epsilon)(1-\alpha)(\alpha+\alpha^2)\beta k$ for some small $\epsilon>0$. Note that there are at most $k'$ target weights to try, and the running time of this part is from Theorem~\ref{thm: SkCC} equal to $2^{k'}\operatorname{poly}(n)$.

We now turn to the correctness of the algorithm. If the first part detects a cycle, it is clearly a $k$-long cycle through $e$ and we are done.
If however the first part fails to find a cycle, we know that a $k$-long cycle through $e$, when it exists in $G$, must have length at least $\beta k$. Let $C=\{s_1,s_2,\cdots, s_\ell\}$  be any such cycle. We can divide the cycle into \emph{runs} of consecutive vertices $s_j,s_{j+1},\cdots,s_{j+l}\in S,s_{j-1},s_{j+l+1}\not \in S$ (called $S$-\emph{runs}) and runs of consecutive vertices not in $S$ in between. Our key observation inspired by the bipartite case, is that 
\begin{observation}
\label{obs: Xrun}
In any $S$-run of vertices $s_j,s_{j+1},\cdots,s_{j+l}\in S,s_{j-1},s_{j+l+1}\not \in S$, there is no mapping from the edges $s_{j-1}s_j,s_js_{j+1},\cdots, s_{j+l}s_{j+l+1}$ to $l+2$ different colours in the sets $c$, because there are only $l+1$ colours available in $\bigcup_{i=j-1}^{j+l} c(s_is_{i+1})$.
\end{observation}
Hence, we want to find a cycle with many $S$-runs as at least one edge per $S$-run cannot be a rainbow edge in the colourful cycle detected by the algorithm. Thus each $S$-run will get us at least one free edge. We note that the number of split rainbow edges provides a lower bound on two times the number of $S$-runs along any cycle detected by the algorithm.
With the chosen parameters $k'=\lceil (1-(1-\epsilon)(1-\alpha)(\alpha+\alpha^2)\beta/2)k \rceil$ and target weight $t\geq (1-\epsilon)(1-\alpha)(\alpha+\alpha^2)\beta k$, we 
will detect a cycle of length at least $k'+\frac{t}{2}\geq k$.

We next show that a long enough cycle in $G$ guarantees that we can find a colourful cycle with many $S$-runs.
\begin{lemma}
\label{lem: colourful}
For any fixed $\epsilon>0$,  any $\beta k$-long cycle $C$ has a set of at least $(1-(1-\epsilon)(1-\alpha)(\alpha))\beta k$ differently coloured edges of total weight at least $(1-\epsilon)(1-\alpha)(\alpha+\alpha^2)\beta k$ with probability at least $1-\exp(-\Omega(k))$.
\end{lemma}
\begin{proof}
Let $C=\{s_1,s_2,\cdots, s_\ell\}$ and define the set $X=X_{\mbox{\tiny left}}\cup X_{\mbox{\tiny right}}$ with $X_{\mbox{\tiny left}}$ consisting of the edges $s_is_{i+1}$ with $s_i\not \in S,s_{i+1}\in S$, and $X_{\mbox{\tiny right}}$ consisting of the edges $s_is_{i+1}$ with $s_{i-1},s_i\in S,s_{i+1}\not \in S$. Note that all edges in $X$ have a unique colour in $c$ (observe the asymmetry in $X_{\mbox{\tiny left}}$ and $X_{\mbox{\tiny right}}$), and that they are all split edges. With $\tau$ denoting the total number of split edges along $C$, we have that $|X|\geq |X_{\mbox{\tiny left}}|=\tfrac{\tau}{2}$. The expected number of edges in $X_{\mbox{\tiny left}}$ is $(1-\alpha)\alpha\ell$, and in $X_{\mbox{\tiny right}}$, $\alpha^2(1-\alpha)\ell$. Together the expected number of edges in $X$ is $(1-\alpha)\alpha\ell+\alpha^2(1-\alpha)\ell\geq  (1-\alpha)(\alpha+\alpha^2)\beta k$.
We want to lower bound the probability that both $|X_{\mbox{\tiny left}}|$ and $|X|$ are close to their respective expectations. To this end, we partition $X_{\mbox{\tiny left}}=\bigcup_{j=0}^1 X_{\mbox{\tiny left}}^{(j)}$ with $ X_{\mbox{\tiny left}}^{(j)}=\{s_is_{i+1}\in  X_{\mbox{\tiny left}}|i\equiv j \operatorname{mod } 2\}$, and  $X_{\mbox{\tiny right}}=\bigcup_{j=0}^2 X_{\mbox{\tiny right}}^{(j)}$ with $ X_{\mbox{\tiny right}}^{(j)}=\{s_is_{i+1}\in  X_{\mbox{\tiny right}}|i\equiv j \operatorname{mod } 3\}$. The reason for the partitioning is that in each $X_{\mbox{\tiny left}}^{(j)}$ and $X_{\mbox{\tiny right}}^{(j)}$ the elements are included with independent probability of each other and their cardinalities can be analysed by Chernoff bounds. In particular, for an arbitrarily small $0<\epsilon<1$ we have
\[
\Pr\left(\left| |X_{\mbox{\tiny left}}^{(j)}|-\mathbb{E}(|X_{\mbox{\tiny left}}^{(j)}|)\right | > \epsilon\mathbb{E}(|X_{\mbox{\tiny left}}^{(j)}| \right)\leq 2\exp(-\tfrac{\epsilon^2\mathbb{E}(|X_{\mbox{\tiny left}}^{(j)}|)}{3})\leq 2\exp(-\tfrac{\epsilon^2\alpha(1-\alpha)\ell}{6}),
\]
and
\[
\Pr\left(\left| |X_{\mbox{\tiny right}}^{(j)}|-\mathbb{E}(|X_{\mbox{\tiny right}}^{(j)}|)\right| > \epsilon\mathbb{E}(|X_{\mbox{\tiny right}}^{(j)}| \right)\leq 2\exp(-\tfrac{\epsilon^2\mathbb{E}(|X_{\mbox{\tiny right}}^{(j)}|)}{3})\leq 2\exp(-\tfrac{\epsilon^2\alpha^2(1-\alpha)\ell}{9}).
\]
Hence, for any fixed $\epsilon$, all these five events' probabilities are of size $\exp(-\Omega(k))$. From a union bound of all of them, the probability that none of these events happen, meaning that both $|X_{\mbox{\tiny left}}| \in (1\pm \epsilon)\alpha(1-\alpha)\ell$ and $|X| \in  (1\pm \epsilon)(1-\alpha)(\alpha+\alpha^2)\ell$, is $1-\exp(-\Omega(k))$.
In the following assume that this is the case.
We will next try to find as many differently coloured edges as possible. We can find different colours for all edges in an $S$-run except for one arbitrarily chosen edge since the edges have their end-points in $S$ as available colours. In particular, we can pick different colours for all but a non-split edge in the $S$-run whenever one exists. Furthermore, these chosen colours are not available for any other edges on $C$ and hence are different from whatever colours other parts of $C$ use. Note also that edges $s_is_{i+1}$ with $s_i,s_{i+1}\not \in X$ have  a unique colour by design, and also that we can assume that at least $(1-\epsilon)(1-\alpha)(\alpha+\alpha^2)\beta k$ edges in $X$ are part of the differently coloured edges selected. The number of $S$-runs on $C$ is $\tau/2$. Together, we can find a set of at least $\ell-\tau/2= \ell-|X_{\mbox{\tiny left}}|\geq \beta k - (1+\epsilon)(1-\alpha)\alpha\beta k$ differently coloured edges along $C$ in $G$ that contains $(1-\epsilon)(1-\alpha)(\alpha+\alpha^2)\beta k$ split edges.
\end{proof}

From the above Lemma~\ref{lem: colourful}, we can conclude that our algorithm will detect a $k$-long cycle in $G$ through $e$ with probability $1-\exp(-\Omega(k))$ times the success probability of the algorithm in Theorem~\ref{thm: SkCC}, provided 
\[
(1-(1+\epsilon)(1-\alpha)\alpha)\beta k\geq k'=\lceil (1-(1-\epsilon)(1-\alpha)(\alpha+\alpha^2)\beta/2)k \rceil.
\]
With the parameters $\alpha=.5774$, $\beta=1.0856$, and $\epsilon=.01$, we satisfy the above and get running time 
\[
\max(1.657^{\beta k},2^{k'})\operatorname{poly}(n)<1.7304^k\operatorname{poly}(n).
\]

\section*{Appendix}
\begin{proof}[Proof of Lemma~\ref{lem: cc2per}.]
  If $D$ is simple (with loops), the result is well-known.
  It is based on the observation that every cycle cover $K$ corresponds to exactly one permutation in Eq.~\eqref{eq: per def},
  namely $\sigma(i)=j$ for the arc $(v_i,v_j)$.
  \medskip

  Otherwise consider a pair $(u,v)$ of distinct vertices and let $\{e_1,\ldots, e_m\}$ be the multiple arcs from $u$ to $v$.
  Every cycle cover uses at most one of these arcs, so we
  can partition the cycle covers 
  $\mathscr C$ of $D$ into $\mathscr C_0, \mathscr C_1,\ldots, \mathscr C_m$ according to which of those $m$ arcs they contain.
  (Here, $\mathscr C_0$ contains the cycle covers that don't use any of them.)
  For $i\in \{1,\ldots, m\}$,
  the contribution of a cycle cover $K_i\in \mathscr C_i$ is then
  \[ \lambda(K_i)= \prod_{e\in E(K_i)} \lambda (e) = 
  \lambda(e_i)\prod_{e\in E(K_i) -\{e_i\}} \lambda(e) = 
  \lambda(e_i)\prod_{e\in\dot E(K_i)} \lambda(e)
  \,,\]
  where we have introduced the notation 
  \[ \dot E(H) =   E(H) - \{ \, e\in E(H) \colon \operatorname{init}(e) = u, \operatorname{ter}(e) =v \,\}\]
  for $H$ with all arcs from $u$ to $v$ removed (for an arc $e=(u,v)$, we write $\operatorname{init}(e)=u$ and $\operatorname{ter}(e)=v$).
  
  We now have an expression for the label sum over all cycle covers:
  \begin{equation}\label{eq: split label sum}
  \sum_{K \in \mathscr C} \lambda(K) = 
  \sum_{i=0}^m 
  \sum_{K_i \in \mathscr C_i} \lambda(K_i) = 
  \sum_{K_0\in \mathscr C_0}\lambda (K_0) +
  \sum_{i=1}^m 
    \lambda(e_i)
    \sum_{K_i \in \mathscr C_i} 
  \prod_{e \in \dot E(K_i)} \lambda(e) \,.
  \end{equation}
  Crucially, the $i$s in the rightmost sum play no role, because
  \[
    \sum_{K_i \in \mathscr C_i} \prod_{e \in \dot E(K_i)} \lambda(e) 
    = 
    \sum_{K_1 \in \mathscr C_1} \prod_{e \in \dot E(K_1)} \lambda(e) \,.
  \]
  To see this, any cycle cover $K_i\in \mathscr C_i$ can be turned into a cycle cover $K_j\in \mathscr C_j$ by changing the arc $e_i$ to $e_j$,
  and the contribution of $K_i$ and $K_j$ to the product (where neither $e_i$ or $e_j$ appears) is the same.
  Thus, we can reassemble the rightmost term of Eq.~\eqref{eq: split label sum} as follows:
  \begin{multline*}
  \sum_{i=1}^m 
    \lambda(e_i)
    \sum_{K_i \in \mathscr C_i} 
    \prod_{e \in \dot E(K_i)} \lambda(e)  =
  \sum_{i=1}^m 
    \lambda(e_i)
    \sum_{K_1 \in \mathscr C_1} 
    \prod_{e \in \dot E(K_1)} \lambda(e) 
    =\\
    \left(\sum_{K_1 \in \mathscr C_1} 
    \prod_{e \in \dot E(K_1)} \lambda(e) \right)
  \sum_{i=1}^m 
    \lambda(e_i)
    =
    \bigl(\lambda(e_1) +\cdots + \lambda(e_m)\bigr)
    \sum_{K_1 \in \mathscr C_1} 
    \prod_{e \in \dot E(K_1)} \lambda(e)
    \,.
  \end{multline*}

  Now construct a new directed graph $D'$ from $D$ by removing $e_1,\ldots, e_m$ and replacing them by a single new arc~$e'$ with label $\lambda(e') = \lambda(e_1)+\cdots + \lambda(e_m)$.
  Clearly, the cycle covers of $D'$ can be partitioned into two sets:
  those that use the arc~$e'$ and those that avoid it.
  The former cycle covers also avoid $\{e_1,\ldots, e_m\}$ in $D$, so they constitute the same set $\mathscr C_0$.
  The latter cycle covers are in one-to-one correspondence with $\mathscr C_1$ (by identifying $e_1\in E(D)$ with $e'\in E(D')$), and the contribution from such a cover $K'$ of $D'$ is the same:
  \[ \prod_{e\in E(K')}  \lambda(e) = 
  \lambda(e_{uv}) 
  \prod_{e\in \dot E(K') } \lambda(e) =
  \bigl(\lambda(e_1) + \cdots +\lambda(e_m)\bigr)
  \prod_{e\in \dot E(K') } \lambda(e)\,.\]
  In particular, $D$ and $D'$ have the same label sum.

  We repeat this process until no multiple arcs remain.
\end{proof}

\ifdefined\A
\else
\subsection*{Acknowledgements}
This work is supported by the VILLUM Foundation, Grant 16582.
\fi

\bibliographystyle{abbrv}
\bibliography{longestcycle}

\begin{thebibliography}{10}

\bibitem{Berkowitz1984}
S.~J. Berkowitz.
\newblock On computing the determinant in small parallel time using a small
  number of processors.
\newblock {\em Information Processing Letters}, 18(3):147--150, 1984.

\bibitem{Bjorklund2014}
A.~Bj{\"{o}}rklund.
\newblock Determinant sums for undirected {H}amiltonicity.
\newblock {\em {SIAM} J. Comput.}, 43(1):280--299, 2014.

\bibitem{BjorklundHKK2017}
A.~Bj{\"{o}}rklund, T.~Husfeldt, P.~Kaski, and M.~Koivisto.
\newblock Narrow sieves for parameterized paths and packings.
\newblock {\em J. Comput. Syst. Sci.}, 87:119--139, 2017.

\bibitem{CyganFKLMPPS2015}
M.~Cygan, F.~V. Fomin, L.~Kowalik, D.~Lokshtanov, D.~Marx, M.~Pilipczuk,
  M.~Pilipczuk, and S.~Saurabh.
\newblock {\em Parameterized Algorithms}.
\newblock Springer, 2015.

\bibitem{DeMilloL1978}
R.~A. Demillo and R.~J. Lipton.
\newblock A probabilistic remark on algebraic program testing.
\newblock {\em Information Processing Letters}, 7(4):193--195, 1978.

\bibitem{Eiben24}
E.~Eiben, T.~Koana, and M.~Wahlström.
\newblock Determinantal sieving.
\newblock In {\em Proceedings of the 2024 Annual ACM-SIAM Symposium on Discrete
  Algorithms (SODA)}, pages 377--423, 2024.

\bibitem{FominGKSS2023_full}
F.~V. Fomin, P.~A. Golovach, T.~Korhonen, K.~Simonov, and G.~Stamoulis.
\newblock Fixed-parameter tractability of maximum colored path and beyond.
\newblock {\em ACM Trans. Algorithms}, June 2024.

\bibitem{FominGSS2022}
F.~V. Fomin, P.~A. Golovach, D.~Sagunov, and K.~Simonov.
\newblock Longest cycle above {E}rd{\H o}s--{G}allai bound.
\newblock In S.~Chechik, G.~Navarro, E.~Rotenberg, and G.~Herman, editors, {\em
  30th Annual European Symposium on Algorithms (ESA 2022)}, volume 244 of {\em
  Leibniz International Proceedings in Informatics (LIPIcs)}, pages
  55:1--55:15, Dagstuhl, Germany, 2022. Schloss Dagstuhl -- Leibniz-Zentrum
  f{\"u}r Informatik.

\bibitem{Karp1972}
R.~M. Karp.
\newblock Reducibility among combinatorial problems.
\newblock In R.~E. Miller and J.~W. Thatcher, editors, {\em Proceedings of a
  symposium on the Complexity of Computer Computations, held March 20-22, 1972,
  at the {IBM} Thomas J. Watson Research Center, Yorktown Heights, New York,
  {USA}}, The {IBM} Research Symposia Series, pages 85--103. Plenum Press, New
  York, 1972.

\bibitem{Koutis2008}
I.~Koutis.
\newblock Faster algebraic algorithms for path and packing problems.
\newblock In L.~Aceto, I.~Damg{\aa}rd, L.~A. Goldberg, M.~M. Halld{\'{o}}rsson,
  A.~Ing{\'{o}}lfsd{\'{o}}ttir, and I.~Walukiewicz, editors, {\em 35th
  International Colloquium on Automata, Languages and Programming, {ICALP}
  2008, July 7-11, 2008, Reykjavik, Iceland}, volume 5125 of {\em Lecture Notes
  in Computer Science}, pages 575--586. Springer, 2008.

\bibitem{Schwartz1980}
J.~T. Schwartz.
\newblock Fast probabilistic algorithms for verification of polynomial
  identities.
\newblock {\em J. ACM}, 27(4):701--717, 1980.

\bibitem{Wahlstrom2013}
M.~Wahlstr{\"{o}}m.
\newblock Abusing the {T}utte matrix: {A}n algebraic instance compression for
  the {K}-set-cycle problem.
\newblock In N.~Portier and T.~Wilke, editors, {\em 30th International
  Symposium on Theoretical Aspects of Computer Science, {STACS} 2013, February
  27 - March 2, 2013, Kiel, Germany}, volume~20 of {\em LIPIcs}, pages
  341--352. Schloss Dagstuhl - Leibniz-Zentrum f{\"{u}}r Informatik, 2013.

\bibitem{Williams2009}
R.~Williams.
\newblock Finding paths of length $k$ in ${O}^*(2^k)$ time.
\newblock {\em Inf. Process. Lett.}, 109(6):315--318, 2009.

\bibitem{Zehavi2016}
M.~Zehavi.
\newblock A randomized algorithm for long directed cycle.
\newblock {\em Inf. Process. Lett.}, 116(6):419--422, 2016.

\bibitem{Zippel1979}
R.~Zippel.
\newblock Probabilistic algorithms for sparse polynomials.
\newblock In E.~W. Ng, editor, {\em Symbolic and Algebraic Computation}, pages
  216--226. Springer Berlin Heidelberg, 1979.

\end{thebibliography}

\end{document}